\documentclass[pra,onecolumn,notitlepage,aps,10pt]{revtex4-1}

\usepackage{amssymb,amsmath,bm,xcolor} 
\usepackage[colorlinks=true,citecolor=blue,urlcolor=blue]{hyperref}
\usepackage{mathrsfs}
\usepackage{cases}
\usepackage{graphicx}
\usepackage{subfigure}
\usepackage{here}
\newcommand{\mc}{\mathcal}
\newcommand{\Tr}{\mathrm{Tr}}
\newcommand{\id}{I}
\usepackage{blindtext}
\usepackage{listings}
\def\<{\langle}\def\>{\rangle}
\newtheorem{theorem}{Theorem}[section]

\newtheorem{defi}[theorem]{Definition}
\newtheorem{proposition}[theorem]{Proposition}

\newenvironment{proof}[1][Proof:]{\begin{trivlist}
\item[\hskip \labelsep {\bfseries #1}]}{\end{trivlist}
\hfill$\square$}

\usepackage[margin=3cm]{geometry}
\usepackage[T1]{fontenc}

\begin{document}
\title{Entropic time-energy uncertainty relations: An algebraic approach}
\author{Christian Bertoni}
\email{chr.bertoni@gmail.com}
\address{Institute for Theoretical Physics, ETH Z\"urich, 8093 Z\"urich, Switzerland}
\author{Yuxiang Yang}
\email{yangyu@phys.ethz.ch}
\address{Institute for Theoretical Physics, ETH Z\"urich, 8093 Z\"urich, Switzerland}
\author{Joseph M. Renes}
\email{joerenes@gmail.com}
\address{Institute for Theoretical Physics, ETH Z\"urich, 8093 Z\"urich, Switzerland}
\begin{abstract}
We address entropic uncertainty relations between time and energy or, more precisely, between measurements of an observable $G$ and the displacement $r$ of the $G$-generated evolution $e^{-ir G}$.
We derive lower bounds on the entropic uncertainty in two frequently considered scenarios, which can be illustrated as two different guessing games in which the role of the guessers are fixed or not.  In particular, our bound for the first game improves the previous result by Coles et al.\ \cite{TEEUR}.
Our derivation uses as a subroutine a recently proposed novel algebraic method \cite{SSA}, which can in general be used to derive a wider class of entropic uncertainty principles.

\end{abstract}

\maketitle

\section{Introduction}
Uncertainty principles are a cornerstone of modern physics~\cite{heisenberg_uber_1927}. The most famous instantiation is perhaps the Kennard relation \cite{kennard1927} $\sigma_x \sigma_p \geq  \hbar/2$
where $\sigma_x$ and $\sigma_p$ are the standard deviations of the measurement of the position and the momentum of a particle respectively. 
Entropic uncertainty relations, in contrast, offer an operational interpretation of the uncertainty principle, which is often more desireable in applications such as quantum cryptography.
The most well-known entropic uncertainty relation was derived by Maassen and Uffink \cite{MU}: Let $\rho$ be the density matrix of a system $A$ and $\mathcal{E}_V$ and $\mathcal{E}_W$ be the measurement quantum channels for the observables $V$ and $W$, then
\begin{equation}\label{mu}
S(A)_{\mathcal{E}_V(\rho)}+S(A)_{\mathcal{E}_W(\rho)}\geq -\log{\max_{k,j}|\langle v_j|w_k\rangle|^2}\,,
\end{equation}
where $|v_i\rangle$ and $|w_i\rangle$ are the eigenvectors of $V$ and $W$ and $S(A)_\rho$ is the von Neumann entropy of the state $\rho$ on system $A$. 
This relation can be interpreted as a guessing game: Alice has the quantum state $\rho$ and can choose whether to measure $V$ or $W$, Bob wins if he can correctly guess the result of the measurement. Equation \eqref{mu} prevents Bob from perfectly winning this game, provided the right hand side is non zero, i.e.\ $V$ and $W$ do not commute.
Indeed, if $S(A)_{\mathcal{E}_V(\rho)}=0$, meaning that he can perfectly guess the measurement result of $V$, then the inequality implies $S(A)_{\mathcal{E}_W(\rho)}\geq-\log{\max_{k,j}|\langle v_j|w_k\rangle|^2}$, and thus Bob will not be able to perfectly guess the measurement result of $W$.

The entropic uncertainty relation in Eq.\ \eqref{mu} has been further extended to account for the effect of quantum memories \cite{renes_conjectured_2009,qmem}:  If  a quantum memory $B$ is entangled with the original system $A$, Bob could use it to deduce Alice's measurement outcomes. There are essentially two possible uses of the memory, corresponding to two guessing games. The first game, also referred to as \emph{the tripartite game}, concerns splitting the quantum memory into two parts $B_1$ and $B_2$, where $B_1$ is used for guessing $V$ and $B_2$ is used for guessing $W$. Then the following entropic uncertainty relation holds \cite{renes_conjectured_2009,qmem}
\begin{equation}\label{uqmem0}
S(A|B_1)_{\mathcal{E}_V(\rho)}+S(A|B_2)_{\mathcal{E}_W(\rho)}\geq -\log{\max_{k,j}|\langle v_j|w_k\rangle|^2}\,,
\end{equation}
where the measurements are performed only on the system $A$  and $S(A|B)_\rho$ is the quantum conditional entropy of $A$ conditioned on $B$.
On the other hand, the second game regards the memory as a whole and is referred to as \emph{the bipartite game}. In this case, the uncertainty relation becomes
\begin{equation}\label{uqmem}
S(A|B)_{\mathcal{E}_V(\rho)}+S(A|B)_{\mathcal{E}_W(\rho)}\geq -\log{\max_{k,j}|\langle v_j|w_k\rangle|^2} +S(A|B)_\rho\,.
\end{equation}
In this case, Bob, who keeps the quantum memory, can increase his chance of winning by referring to it. 
In fact, since the quantum conditional entropy can be negative, Bob can win the game with certainty by using a suitable entangled state for which the right hand side of Eq.\ \eqref{uqmem} vanishes.

The two guessing games differ only in whether the memory is split into two parts or not.
This difference highlights a subtlety of the uncertainty principle that it is impossible to simultaneously know the values of two noncommuting observables of the same system. 
On the one hand, by splitting the memory, it is possible to provide guesses for both observables at the same time. 
The fact that the tripartite game cannot be won then matches the uncertainty principle. 
On the other hand, in each round of the bipartite game Bob only has to guess one of the observables. Therefore, 
using a quantum memory can allow him to win the game with certainty, in seeming contravention of the uncertainty principle.


Various extensions of these entropic uncertainty relations with memory have been put forward  [see, e.g., Refs.\ \cite{tomamichel2011uncertainty,Coles2011Information,coles2012uncertainty,furrer2014position} and Ref.\ \cite{survey} for a full survey]. 
A natural question is whether there is an entropic time-energy uncertainty relation. 
This is a more subtle situation than relations involving measurements of observables, since an ideal time observable does not exist for finite dimensional systems \cite{pauli,pauli2012general,susskind1964quantum}. Possible ways out include defining an approximate time operator \cite{rastegin2019entropic}, or considering the uncertainty of measuring the duration of evolutions, i.e.\ measuring the state as a quantum clock, instead of directly measuring time.

In this work, we take the latter approach and study the tradeoff between uncertainties of measuring an observable $G$ (e.g.\ the Hamiltonian of the system) and determining a parameter $r$ of the unitary evolution $e^{-ir G}$.
Unlike most of the previous works, whose proofs are built on basic properties of quantum entropies and distances, we take a new algebraic approach that makes use of a strong subadditivity on algebras, developed recently by Gao, Junge, and Laracuente \cite{SSA}.
As a result, we obtain entropic uncertain relations for both of the aforementioned guessing games.
Entropic time-uncertainty relations were recently studied in the setting of the tripartite guessing game by Coles et al.\ \cite{TEEUR}.
In comparison, we show that our bound is strictly tighter than their result for von Neumann entropies, though they also study more general R\'enyi entropies. 
 
The rest of the paper is arranged as follows. In Section \ref{gamesandrel}, we define the two guessing games under consideration and state our main results on the entropic uncertainty relation. In Section \ref{prelim}, we prepare for the proofs of the uncertainty relations by introducing a few useful results from Ref.\ \cite{SSA}. In Section \ref{rot}, we prove our bounds on the entropic uncertainties. In Section \ref{numerical}, we present some numerical examples that show the tightness and advantage of our results. Finally, in Section \ref{conclusions}, we conclude with a few discussions.

\section{Guessing games and entropic uncertainty relations}\label{gamesandrel}
In this section, we introduce the setting and the main results of our paper. 
Entropic uncertainty relations arise naturally from guessing games, where players are asked to make guesses on random operations performed by an extra player.
We will propose here two different  guessing games that lead to different entropic uncertain relations.

We focus on  guessing games involving a game operator $A$ and one or multiple guessers, where the operation performed by $A$ is either a measurement of an observable $G$ or a rotation $\rho\mapsto e^{-iGr_k}\rho e^{iGr_k}$ generated by $G$ with $r_k$ being a random number drawn from a fixed finite set $\{r_k\}_{k=1}^{|R|}$.  

Now we are ready to introduce the first guessing game: 
\begin{defi}[The tripartite guessing game] The game concerns two guessers $B_1$ and $B_2$ and runs as follows:
	\begin{enumerate}
		\item[0.] (Setup) Three players $A$, $B_1$, and $B_2$ share a quantum state $\rho_{AB_1B_2}$, fix a probability distribution $\{p_k\}_{k=1}^{|R|}$, a generator $G$ acting on $A$, and a set of rotations $\{r_k\}_{k=1}^{|R|}$.
		\item[1.] $A$ tosses a coin to choose between measuring $G$ or applying a rotation $e^{-iGr_k}$.\footnote{Since the rotation does not affect the measurement, we could also say that $A$ always applies a random rotation and then randomly chooses whether to measure $G$. This version is more easily interpretable if one wants to consider time evolution as the rotation.}
		\item[2.a] If $A$ gets a head, she chooses an $r_k$ following a probability distribution $\{p_k\}_{k=1}^{|R|}$ and applies $e^{-iGr_k}$ to her part of $
		\rho$. She then sends the rotated state to $B_1$, with instructions to guess $r_k$.  
		\item[2.b] If $A$ gets a tail, she measures $G$ on her part of $\rho$ and asks $B_2$ to guess the measurement outcome.
		\item[3.] Accordingly, $B_1$ or $B_2$ provides his guess. 
	\end{enumerate}
	\end{defi}
A graphical illustration of this game is portrayed in Figure \ref{drawing_first_game}.
\begin{figure}
	\centering
	\includegraphics[width=70mm]{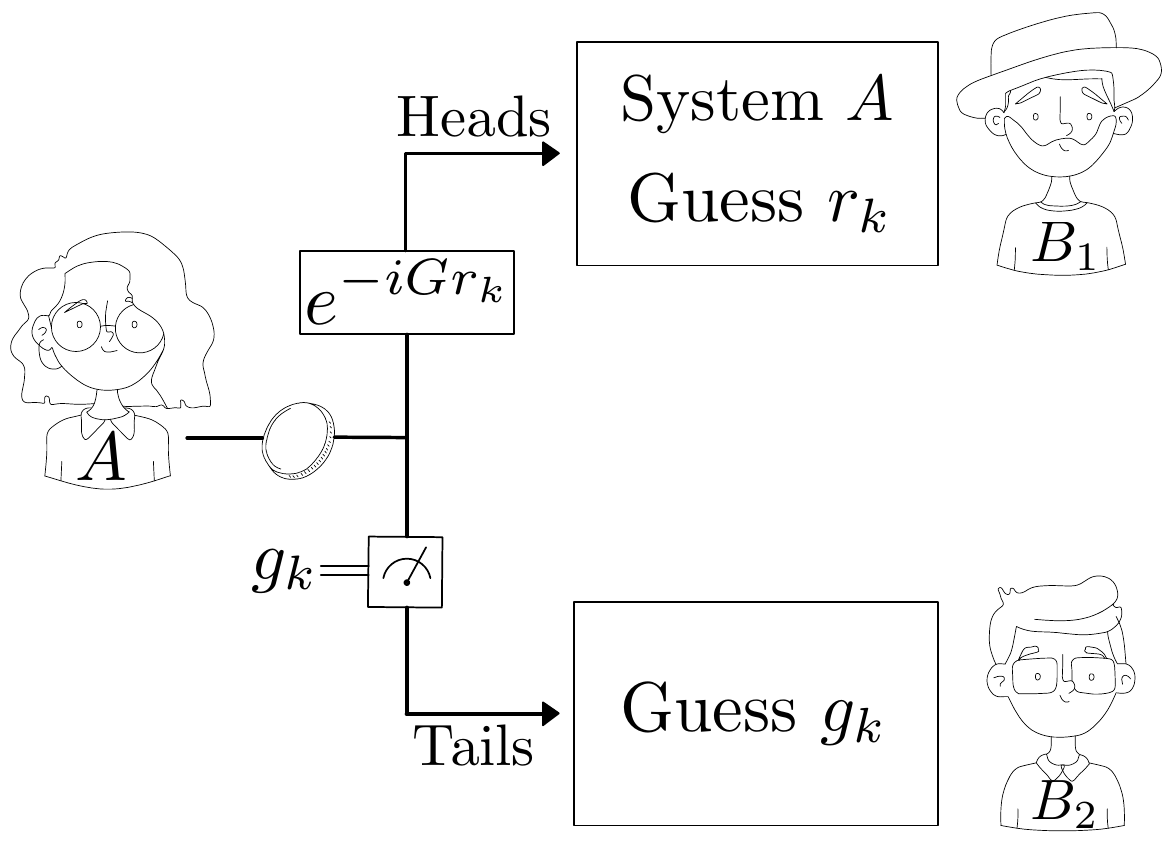}
	\caption{{\bf The tripartite guessing game. } The above figure illustrates the setting of the tripartite guessing game, where two guessers $B_1$ and $B_2$ are assigned different tasks. Depending on the outcome of a coin toss, Alice asks either $B_1$ to guess a rotation or $B_2$ to guess a measurement outcome.}
		\label{drawing_first_game}
\end{figure}
To quantify the uncertainty of the guesses in the above game, we use an ancillary Hilbert space $\mathcal{H}_R$ for the random number $\{r_k\}$, which has probability distribution $\{p_k\}$. 
If $A$ chooses to perform the rotation, the state afterwards is
\begin{equation}\label{kappa}
\kappa_{RAB_1B_2}=\sum_{k=1}^{|R|} p_k|r_k\rangle\langle r_k| \otimes e^{-iGr_k}\rho_{AB_1B_2}e^{iGr_k}\,.
\end{equation}
If $A$ chooses to measure $G$, the state afterwards is 
\begin{equation}\label{omega}
\omega_{AB_1B_2}=\sum_{k=1}^{|A|}|g_k \rangle \langle g_k|\langle g_k| \rho_{AB_1B_2}|g_k \rangle\,,
\end{equation}
where $\{|g_k\>\}$ are the eigenstates of $G$ with eigenvalue $g_k$. 
The quantity
\begin{equation}
S(R|AB_1)_\kappa+S(A|B_2)_\omega\,,
\end{equation}
represents the total uncertainty of the game, in the sense that the larger it is, the more difficult it is to guess correctly.

Our first result is a lower bound of the total uncertainty, as described in the following theorem.
\begin{theorem}
The total uncertainty of the tripartite game is lower bounded as
\begin{equation}\label{strel}
S(R|AB_1)_\kappa+S(A|B_2)_\omega  \geq S(R)_\kappa+ D(\kappa_{AB_1} || \omega_{AB_1}) + \max\{0,I(A:B_1)_\omega-I(B_1:B_2)_\rho+S(A|B_1B_2)_\rho\}\,.
\end{equation}
The bound is saturated if $\rho_{AB_1B_2}$ is pure or $\rho_{AB_1B_2}=\rho_{AB_1}\otimes\rho_{B_2}$.
\end{theorem}

Our bound (\ref{strel}) manifests a tradeoff relation between guessing the measurement outcome and guessing the rotation. 
In particular, it shows that it is impossible for both guesses to be perfect for the same state (unless $R$ is trivial), since the right hand side of the bound (\ref{strel}) is always positive. If the conditional entropy $S(A|B_2)_\omega$ is really low, meaning that $B_2$ can easily guess the measurement value, then the entropy of the rotation chosen must be large to satisfy the bound, making it hard for $B_1$ to guess precisely which rotation has been applied.

Note that, in the case $p_k=\frac{1}{|R|}$ for all $k$, the term $S(R)_\kappa$ is simply $\log |R|$.
 Clearly, to minimize the uncertainty, $B_1$ and $B_2$ want to reduce the last term in the bound (\ref{strel}). From this we can deduce the following conditions for making the uncertainty small:
\begin{itemize}
	\item $B_1$ and $B_2$ need to be as correlated as possible so as to maximize $I(B_1:B_2)_\rho$.
	\item The system $B_1$, which is used to guess the rotation, should be as uncorrelated as possible with the measurement result so as to minimize $I(A:B_1)_\omega$.
	\item $A$ and $B_1B_2$ should be entangled so that $S(A|B_1B_2)_\rho$ is negative.
	\end{itemize}

The guessing game proposed by  Coles et al.\ \cite{TEEUR} is a special case of the tripartite game presented here.  
They showed in \cite[Eq.\ (8)]{TEEUR} that when the distribution over  $R$ is uniform, the total uncertainty can be bounded as
\begin{equation}\label{relcoles}
  S(R|AB_1)_\kappa+S(A|B_2)_\omega \geq \log{|R|}\,.
\end{equation}
Furthermore, for $B_1=\mathbb{C}$ is trivial and $B_2=B$, they find a stronger bound in \cite[Eq.\ (E10)]{TEEUR}:
\begin{equation}\label{relcolesold}
S(R|A)_\kappa+S(A|B)_\omega \geq S(R)_\kappa +D(\kappa_{A}||\omega_{A})\,,
\end{equation}
which is tight if $\rho^{AB}$ is pure.  It is clear that our bound (\ref{strel}) is tighter since the additional term $\max\{0,I(A:B_1)_\omega-I(B_1:B_2)_\rho+S(A|B_1B_2)_\rho\}$ is positive.

In the first game, the system $B$ is broken into two subsystems $B_1$ and $B_2$ and distributed to individual players, whose tasks are fixed. Alternatively, we can consider a variation of the game where $B$ is given to a single player, who may be given either task (to guess the measurement outcome or the rotation).
\begin{defi}[The bipartite guessing game] The game concerns only one guesser $B$ and runs as follows:
	\begin{enumerate}
		\item[0.] (Setup) Two players $A$ and $B$ share a quantum state $\rho_{AB}$, fix a probability distribution $\{p_k\}_{k=1}^{|R|}$, a generator $G$ acting on $A$, and a set of rotations $\{r_k\}_{k=1}^{|R|}$.
		\item[1.] $A$ tosses a coin to choose between measuring $G$ or applying a rotation $e^{-iGr_k}$. 
		\item[2.a] If $A$ gets a head, she chooses an $r_k$ following a probability distribution $\{p_k\}_{k=1}^{|R|}$ and applies $e^{-iGr_k}$ to her part of $
		\rho$. She then sends the rotated state to $B$, with instructions to guess $r_k$.  
		\item[2.b] If $A$ gets a tail, she measures $G$ on her part of $\rho$ and asks $B$ to guess the measurement outcome.
		\item[3.] $B$ provides his guess.
	\end{enumerate}
	\end{defi}
A graphical illustration of this game is portrayed in Figure \ref{drawing_second_game}.
\begin{figure}[h]
	\centering
	\includegraphics[width=70mm]{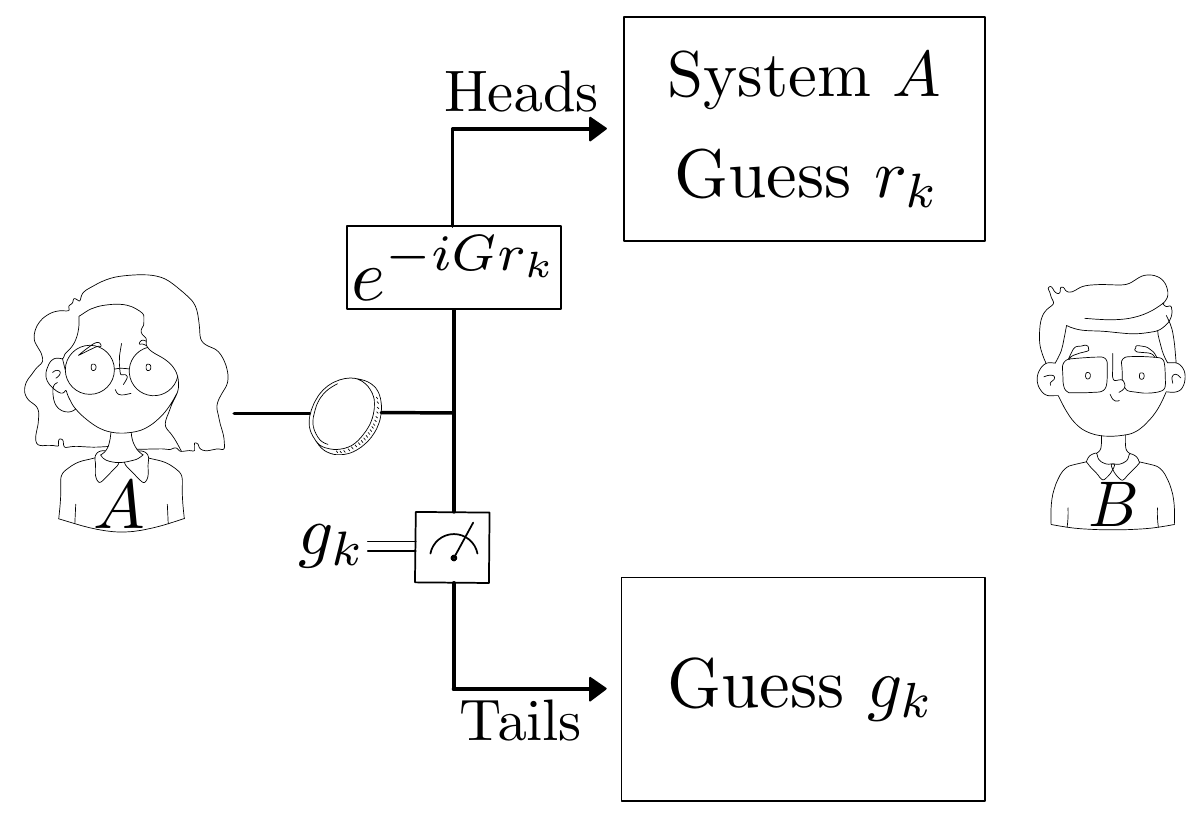}
	\caption{{\bf The bipartite guessing game. } The above figure illustrates the setting of the bipartite guessing game, where the guesser may be asked to guess either a rotation or a measurement outcome.}
	\label{drawing_second_game}
\end{figure}

In this game the quantity that characterizes the uncertainty is 
\begin{equation}
S(R|AB)_\kappa+S(A|B)_\omega\,,
\end{equation}
where $\kappa$ and $\omega$ are defined by Eqs.\ (\ref{kappa}) and (\ref{omega}), respectively.
Just as the tripartite game, we can bound this total uncertainty as well.
\begin{theorem}
The total uncertainty for the bipartite game is lower bounded as
\begin{equation}\label{eub-second}
S(R|AB)_\kappa+S(A|B)_\omega\geq S(R)_\kappa+ D(\kappa_A||\omega_A)+S(A|B)_\rho\,.
\end{equation}
The bound is saturated if $\rho_{AB}=\rho_A\otimes\rho_B$ is a product state or if $\rho_A$ is a pure eigenstate of $G$. 
\end{theorem}
An intriguing distinction between this bound and the bound for the tripartite game (\ref{strel}) is that $B$ may be able to always guess correctly.
This is analogous to the bound for the uncertainty principle in the presence of quantum memory \cite{qmem}, in the sense that quantum correlations that make $S(A|B)_\rho$ negative can reduce the bound  (\ref{eub-second}) to zero. 
To see this, let us consider a simple example in which Alice and Bob hold a qubit each and the two qubits are in the maximally entangled state. 
Furthermore, take $G=\sigma_z$, $|R|=2$ and the uniform distribution for the rotations. In this case $\kappa_{AB}=\omega_{AB}=\frac{1}{2}(|00\rangle\langle 00| +|11\rangle\langle 11|)  $.  Then clearly the right hand side is $0$ as the relative entropy is $0$ and $S(A|B)_\rho=-1$. Moreover one may verify that $S(RAB)_\kappa=1$ and thus the left hand side is also $0$. Intuitively, in this case the rotations have the same effect of a $\sigma_z$ measurement, and Bob can apply the same strategy in both cases. It also means that it is necessary for Bob to use entanglement to win the game: the game is impossible to win perfectly using only a classical memory.

\section{Preliminary: a general framework for entropic uncertainty relations}\label{prelim}
In this section, we introduce part of the main results of Ref.\ \cite{SSA} that will be used in our proof. 
\subsection{Commuting squares and uncertainty relations}
Let $M$ be an algebra of observables and let a $N\subset M$ be subalgebra. For instance, $M$ may be the observables on a bipartite system, and $N$ the observables on just one system. 
The \emph{conditional expectation} onto $N$ is the unique surjective CPTP and unital map $\mathcal{E}_N :M\to N$ such that for all $\rho\in M, \sigma \in N$
\begin{equation}
\Tr(\sigma\mathcal{E}_N(\rho))=\Tr(\sigma \rho)\,.
\end{equation}

	Given a state $\rho\in M$, 
	the \emph{asymmetry measure} of $\rho$ with respect to $N$ is defined as
	\begin{equation}
	D^N(\rho):=\inf_{\sigma \in s(N)} D(\rho || \sigma)\,,
	\end{equation}
	where $D(\cdot||\cdot)$ is the relative entropy and $s(N)$ denotes the states on $N$.
When $N$ is the image of a conditional expectation $	\mc{E}_N$, we have  
\begin{equation}\label{distancemeasure}
	D^N(\rho)=D(\rho||\mathcal{E}_N(\rho))=S(N)_{\mathcal{E}_N(\rho)}-S(M)_\rho\,,
\end{equation} where $S$ is the von Neumann entropy.
We remark that $D^N$, albeit not a distance measure, captures the distinction between $N$ and $M$.
\begin{defi} [Commuting square]
	A set of four observable algebras satisfying the inclusions
\begin{equation}
\label{eq:commutingsquaresymbol}
\begin{pmatrix}
N &\subset & M\\
\cup & & \cup \\
R & \subset & T
\end{pmatrix}
\end{equation}
is called a commuting square if the conditional expectations satisfy

\begin{equation}
\mathcal{E}_N \circ \mathcal{E}_T =\mathcal{E}_T \circ \mathcal{E}_N = \mathcal{E}_R\,.
\end{equation}
\end{defi}
The following theorem will be the core of our proof, which says that one entropic uncertainty relation can be identified from each commuting square.
\begin{theorem}\label{commsq} Let 
$N,M,R,T$ form a commuting square as in \eqref{eq:commutingsquaresymbol}. Then for all $\rho \in M$

\begin{equation}
S(N)_{\mathcal{E}_N(\rho)}+S(T)_{\mathcal{E}_T(\rho)}\geq S(M)_\rho + S(R)_{\mathcal{E}_R(\rho)}\,,
\end{equation}
which is equivalent to 
\begin{equation}
D^N(\rho)+D^T(\rho)\geq D^R(\rho)\,.
\end{equation}
The relation is saturated if and only if there exists a CPTP map $\mathcal{R}$ such that

\begin{equation}\label{comm-uncertainty1}
\mathcal{R}(\mathcal{E}_N(\rho))=\rho \qquad \mathcal{R}(\mathcal{E}_R(\rho))=\mathcal{E}_T(\rho)\,.
\end{equation}
or equivalently 
\begin{equation}\label{comm-uncertainty2}
\mathcal{R}(\mathcal{E}_T(\rho))=\rho \qquad \mathcal{R}(\mathcal{E}_R(\rho))=\mathcal{E}_N(\rho)\,.
\end{equation}
\end{theorem}
Eqs.\ (\ref{comm-uncertainty1}) and (\ref{comm-uncertainty2}) are uncertainty relations with respect to a commuting square, which we will use to derive bounds on the time-energy uncertainty.

\subsection{Examples of conditional expectations}
We provide here some examples of conditional expectations that will be useful later. From now on, Latin uppercase letters will be used to refer to the algebra of Hermitian operators on a corresponding Hilbert space.

\subsubsection{Embedding}\label{embedding}
Let $AB$ be the algebra of Hermitian operators on $\mathcal{H}_A\otimes \mc{H}_B$. We want to find a conditional expectation that takes us to the algebra $B$. One may notice that the partial trace is not a conditional expectation, as it is not unital. To solve this problem, following Example 2.2 in \cite{SSA}, instead of embedding $B\subset AB$ we embed $I_A\otimes B\subset AB$ where $I_A\simeq \mathbb{C}$ is the algebra generated by $\{cI_A: c\in \mathbb{C}\}$. The embedding is done by the map

\begin{equation}
\mc{T}_A(\rho_{AB})=\frac{1}{|A|}I_A\otimes\rho_B\,,
\end{equation}
where $\rho_B=\Tr_A[\rho_{AB}]$.
The map is clearly unital and CPTP. Let $\sigma = c\id_A\otimes \sigma_B \in I_A\otimes B$ and $\rho_{AB}\in AB$, moreover let $\{|a_k\rangle\}_{k=1}^{|A|}$ be a basis of $\mc{H}_A$. We have

\begin{equation}
\begin{aligned}
\Tr[\sigma \rho_{AB}]&=\Tr[c\id_A\otimes \sigma_B\rho_{AB}]\\&=c\Tr_B\left[\sum_k\sum_j \langle a_k|\left(|a_j\rangle\langle a_j| \otimes \sigma_B\right) \rho_{AB}|a_k\rangle\right]\\&=c\Tr_B\left[\sum_k\langle a_k|_A\sigma_B \rho_AB|a_k\rangle\right]=c\Tr_B[\sigma_B\rho_B]\,.
\end{aligned}
\end{equation}
On the other hand
\begin{equation}
\begin{aligned}
\Tr[\sigma \mc{T}_A(\rho_{AB})]&=\Tr\left[\left(c\id_A\otimes \sigma_B\right)\left(\frac{1}{|A|}\id_A\otimes \rho_B\right)\right]\\&=\frac{c}{|A|}\Tr[\id_A \otimes \sigma_B\rho_B]=c\Tr_B[\sigma_B\rho_B]\,.
\end{aligned}
\end{equation}
\subsubsection{Pinching}
Let $G$ be an observable with full support on $\mc{H}_A$ and $\{|g_k\rangle\}_{k=1}^{|A|}$ be the eigenbasis of $G$. The pinching map
\begin{equation}\label{pinching}
\mathcal{P}_G:\rho_A\mapsto \sum_{k=1}^{|A|} |g_k \rangle\langle g_k| \langle g_k|\rho_A|g_k\rangle
\end{equation}
is a conditional expectation onto $\mathrm{span}\{|g_k \rangle\langle g_k|\}_{k=1}^{|A|}$. Notice that this is also an algebra, consisting of all diagonal elements in $A$, and from now on we denote this kind of subalgebras by $\tilde{A}$. 

It is clear that the pinching map $\mc{P}_G:A\to\tilde{A}$ is unital and CPTP, and for $\sigma=\sum_{k=1}^{|A|} p_k |g_k\rangle\langle g_k|$ we have
\begin{equation}
\Tr(\sigma \mathcal{P}_G(\rho_A))=\sum_{k=1}^{|A|}p_k  \langle g_k|\rho_A|g_k\rangle
\end{equation}
and 
\begin{equation}
\Tr(\sigma \rho_A)=\sum_{k=1}^{|A|}\Tr(p_k |g_k\rangle\langle g_k| \rho_A)=\sum_{k,j=1}^{|A|}\langle g_j|p_k |g_k\rangle\langle g_k|\rho_A|g_j\rangle=\sum_{k=1}^{|A|}p_k  \langle g_k|\rho_A|g_k\rangle\,.
\end{equation}
Therefore, $\mc{P}_G$ is a conditional expectation on the subalgebra $\tilde{A}$ that is diagonal with respect to the eigenbasis of $G$.

\section{Proof of rotation-measurement uncertainty relations}\label{rot}
\subsection{The tripartite game}\label{secrel}
Here we prove the bound (\ref{strel}) on the entropic uncertainty in the tripartite game, where the guesser is supposed to guess both the energy and the rotation at the same time. 
The intuition is to find a commuting square of the following structure:
\begin{equation*}
\begin{pmatrix}
{\rm energy} &\subset & {\rm total}\\
\cup & & \cup \\
{\rm minimum} & \subset & {\rm time}
\end{pmatrix}\,,
\end{equation*}
Here ``time'' or ``energy" refers to a subalgebra of ``total'' whose distance to the ``total'' algebra is given by the time/energy uncertainty, and ``minimum'' is the intersection of ``time'' and ``energy'', determined by the conditional expectation.

In fact, we will find two distinct commuting squares as such. For each commuting square, we derive an independent bound on the entropic uncertainty relation of the tripartite game. Combining the two obtained bounds yields the stronger bound in Eq.\ (\ref{strel}).
\subsubsection{The first bound.}
The following proposition, stated and proved for quantum Rényi entropies in Ref.\ \cite{TEEUR}, will be useful.
\begin{proposition}\label{compl}Let $\rho_{M}$ be a state and $\mc{E}_{N}$ be a conditional expectation, then
\begin{equation}
D^{N}(\rho)=-S(E|M)_{U\rho U^\dagger}\,,
\end{equation}
where $U$ is a Stinespring dilation of $\mc{E}_{N}$ on $ME$. 
\end{proposition}
\begin{proof} Eq.\ \eqref{distancemeasure} states that
\begin{equation}
D^{N}(\rho)=S(N)_{\mc{E}_{N}(\rho)}-S(M)_\rho\,.
\end{equation}
Clearly $S(M)_{U\rho U^\dagger}=S(N)_{\mc{E}_{N}(\rho)}$ as $U$ is a Stinespring dilation of $\mc{E}_{N}$. Moreover, conjugation by an isometry preserves the eigenvalues, we have $S(M)_\rho=S(ME)_{U\rho U^\dagger}$. Combining both equalities, we have
\begin{equation}
D^{N}(\rho)=S(M)_{U\rho U^\dagger}-S(ME)_{U\rho U^\dagger}=-S(E|M)_{U\rho U^\dagger}\,.
\end{equation}
\end{proof}

Let $\mc{H}_R$ be a register to store the parameter of rotation, namely that, if the state of $R$ is $\sum_k p_k|r_k\rangle\langle r_k|$, Alice will perform the rotation $e^{-ir_kG}$ with probability $p_k$. Here  $\tilde{R}$ is the diagonal subalgebra of $R$ with respect to the observable $\sum_k r_k|r_k\>\<r_k|$, and $\tilde{G}$ is the diagonal subalgebra of $A$ with respect to the observable $G$.
With this convention in mind, let us now consider the following commuting square 
\begin{equation*}
\begin{pmatrix}
\tilde{R}AB_1 &\subset & RAB_1\\
\cup & & \cup \\
\tilde{R}\tilde{A}B_1 & \subset & R\tilde{A}B_1
\end{pmatrix}\,,
\end{equation*}
where the conditional expectations are the simply corresponding pinching [see Eq.\ (\ref{pinching})]. 
For any state $\rho_{AB_1B_2}$, we define  $\phi_{RAB_1}=|\Omega\rangle\langle \Omega|_R \otimes \rho_{AB_1}$ with $|\Omega\rangle =\sum_k \sqrt{p_k}|r_k\rangle$. 

Now, let us consider the uncertainty relation of the state 
\begin{equation}
\psi_{RAB_1}= \sum_{k,j=1}^{|R|} \sqrt{p_kp_j}|r_k \rangle \langle r_j| \otimes e^{-iGr_k}\rho_{AB_1}e^{iGr_j}\,,
\end{equation}
obtained by applying the unitary $U=\sum_{k=1}^{|R|}|r_k\rangle\langle r_k|\otimes e^{-iGr_k}$ to $\phi_{RAB_1}$.
The conditional expectations result in the states
\begin{equation}
\begin{aligned}
\psi_{\tilde{R}AB_1}=&\sum_{k=1}^{|R|} p_k|r_k\rangle\langle r_k| \otimes e^{-iGr_k}\rho_{AB_1}e^{iGr_k}=\kappa_{RAB_1}\,,\\
\psi^{R\tilde{A}B_1}=&\sum_{k,j=1}^{|R|}\sum_{l=1}^{|A|}\sqrt{p_kp_j}|r_k\rangle\langle r_j| \otimes e^{-ig_l (r_k-r_j)}|g_l\rangle\langle g_l|\langle g_l|\rho_{AB_1}|g_l \rangle\,\text{, and}\\
\psi^{\tilde{R}\tilde{A}B_1}=&\sum_{k=1}^{|R|} p_k |r_k\rangle \langle r_k| \otimes \sum_{l=1}^{|A|} |g_l\rangle\langle g_l|\langle g_l|\rho_{AB_1}|g_l \rangle=\kappa_R \otimes \omega_{AB_1}\,.
\end{aligned}
\end{equation}
For a register $C$ and an arbitrary state $\rho$ on it, let $\mc{Q}_{C,\rho}$ be the discard and reprepare map
\begin{equation}
\begin{aligned} 
\mc{Q}_{C,\rho}(\sigma_C)= \rho_C
\end{aligned}
\end{equation}
that resets the register's state to $\rho$.
We have 
\begin{equation}\label{recovery}
\begin{aligned}
&U\mc{Q}_{AB_1,\rho}(\psi_{\tilde{R}\tilde{A}B_1}) U^\dagger=\psi_{\tilde{R}AB_1}\\
&U\mc{Q}_{AB_1,\rho}(\psi_{R\tilde{A}B_1}) U^\dagger=\psi_{RAB_1}\,.
\end{aligned}
\end{equation}
Therefore, $\mc{R}(\cdot):=U\mc{Q}_{AB_1,\rho}(\cdot) U^\dagger$ constitutes a valid recovery map. By Theorem \ref{commsq}, we have  
\begin{equation}\label{relraw}
D^{\tilde{R}AB_1}(\psi)+D^{R\tilde{A}B_1}(\psi)= D^{\tilde{R}\tilde{A}B_1}(\psi)\,.
\end{equation}
The following isometry is a Stinespring dilation on $AE$ of the pinching map on $A$
\begin{equation}
V=\sum_{k=1}^{|A|} |g_k\rangle_E \otimes|g_k\rangle\langle g_k|_A\,.
\end{equation}
Proposition\ \ref{compl} applied to the second term of Eq.\ (\ref{relraw}) yields the term $-S(E|RAB_1)_{V\rho V^\dagger}$. Consider a purification $\rho_{AB_1B_2B'}$ of $\rho_{AB_1B_2}$ (if $\rho_{AB_1B_2}$ is already pure $B'$ is trivial), using the duality of conditional entropy one gets  $-S(E|RAB_1)_{V\rho V^\dagger}=S(E|B_2B')_\omega$, with $\omega$ the state in Eq.\ (\ref{omega}). Since the complementary channel of pinching under the Stinespring dilation $V$ is also the same pinching, which means $\omega_E=\omega_A$, and thus $S(E|B_2B')_\omega=S(A|B_2B')_\omega$. Using Eq.\ \ref{distancemeasure} on the two remaining terms one gets for  $\rho_{AB_1B_2}$, we obtain
\begin{equation}
S(\tilde{R}AB_1)_\kappa + S(A|B_2B')_\omega = S(\tilde{R})_\kappa + S(\tilde{A}B_1)_\omega\,.
\end{equation} 
Abandoning the notation where one keeps track of which subalgebra the state is in for the more standard one and subtracting $S(AB_1)_\kappa$ from both sides, the relation becomes
\begin{equation}
S(R|AB_1)_\kappa+S(A|B_2B')_\omega = S(R)_\kappa +S(AB_1)_\omega -S(AB_1)_\kappa\,.
\end{equation} 
Since the pinching $\mc{P}_G$ (as the conditional expectation) on $\kappa_A$ yields $\omega_A$, Eq.\ \eqref{distancemeasure} implies 
\begin{equation}\label{relAE}
S(AB_1)_\omega -S(AB_1)_\kappa=D(\kappa_{AB_1}||\omega_{AB_1})\,,
\end{equation} 
and thus we can express the entropic uncertainty as
\begin{equation}\label{39-inter}
  S(R|AB_1)_\kappa+S(A|B_2B')_\omega= S(R)_\kappa + D(\kappa_{AB_1} || \omega_{EB_1})\,.
\end{equation}
Finally, using the strong subadditivity $S(A|B_2B')_\omega \leq S(A|B_2)_\omega $, we obtain the following bound on the entropic uncertainty
\begin{equation}\label{39}
 S(R|AB_1)_\kappa+S(A|B_2)_\omega\ge S(R)_\kappa + D(\kappa_{AB_1} || \omega_{EB_1})\,. 
\end{equation} 
From Eq.\ (\ref{39-inter}) it is immediate that the equality holds if and only if $I(A:B'|B_2)_\omega = 0 $, which is satisfied when $\rho_{AB_1B_2}$ is pure.

Notice that our bound (\ref{39}) holds for arbitrary $B_1$ and $B_2$, and any arbitrary state of $R$ (i.e. the distribution of the rotation parameter $\{r_k\}$ can be non-uniform). On the other hand, the previous result by Coles et al.\ \cite{TEEUR}, given by Eq.\ (\ref{relcoles}), does not have the second term on the right hand side of Eq.\ (\ref{39}) and assumes $R$ to have a uniform distribution.

\subsubsection{The second bound} 
Let us now consider an alternative commuting square:
\begin{equation}
\begin{pmatrix}
\tilde{R}AB_1\id_{B_2} &\subset &RAB_1B_2\\
\cup & & \cup \\
\tilde{R}\tilde{A}\id_{B_1B_2} & \subset & R\tilde{A}\id_{B_1}B_2
\end{pmatrix}.
\end{equation}
We start from the same state as before, namely
\begin{equation}
\psi_{RAB_1B_2}= \sum_{k,j=1}^{|R|} \sqrt{p_kp_j}|r_k \rangle \langle r_j| \otimes e^{-iGr_k}\rho_{AB_1B_2}e^{iGr_j}\,.
\end{equation}
For the new commuting square, using the uncertainty relation (\ref{comm-uncertainty1}), we get the relation
\begin{equation}
S(RAB_1)_\kappa + S(R\tilde{A}B_2)_\omega \geq S(RAB_1B_2)_\psi + S(\tilde{R}\tilde{A})_\omega\,.
\end{equation}

Notice that $\psi_{R\tilde{A}B_2}=U\left(|\Omega\rangle\langle\Omega|_R\otimes \sum_{k=1}^{|A|}\langle g_k|\rho_{AB_2}|g_k\rangle |g_k\rangle\langle g_k|\right) U^\dagger$ with $U=\sum_{k=1}^{|R|}|r_k\rangle\langle r_k|\otimes e^{-iGr_k}$, hence  $S(R\tilde{A}B_2)_\psi=S(\tilde{A}B_2)_\omega$. Similarly $\psi_{RAB_1B_2}=U\left(|\Omega\rangle\langle\Omega|_R\otimes \rho_{AB_1B_2} \right) U^\dagger$, thus $S(RAB_1B_2)=S(AB_1B_2)$. Moreover $\psi_{\tilde{R}\tilde{A}}$ is a product state.  Hence by subtracting $S(A)_\kappa+S(B_1B_2)_\omega$ from both sides and changing the notation like before
\begin{equation}\label{eqsecrel}
S(R|AB_1)_\kappa + S(A|B_2)_\omega\geq S(R)_\kappa + S(AB_1B_2)_\rho + S(A)_\omega-S(AB_1)_\kappa -S(B_2)_\rho\,.
\end{equation}
To have a better comparison with \eqref{39} we can write, using \eqref{relAE}

\begin{equation}
\begin{aligned} &S(AB_1B_2)_\rho + S(A)_\omega-S(AB_1)_\kappa -S(B_2)_\rho\\
=&D(\kappa_{AB_1} || \omega_{AB_1})+S(AB_1B_2)_\rho  -S(B_2)_\rho+S(A)_\omega-S(AB_1)_\omega\\
=&D(\kappa_{AB_1} || \omega_{AB_1}) +I(A:B_1)_\omega-I(B_1:B_2)_\rho+S(A|B_1B_2)_\rho\,.
\end{aligned}
\end{equation}
We can combine this with the previous relation and get, as promised
\begin{equation}\label{relA}
 S(R|AB_1)_\kappa+ S(A|B_2)_\omega \geq  S(R)_\kappa+ D(\kappa_{AB_1} || \omega_{AB_1}) + \max(0,I(A:B_1)_\omega-I(B_1:B_2)_\rho+S(A|B_1B_2)_\rho)\,.
\end{equation}
If $\rho_{AB_1B_2}$ is pure this bound simply reduces to the previous one: as a matter of fact in this case the new term vanishes, since:
\begin{equation}
S(AB_1B_2)_\rho-S(B_2)_\rho+S(A)_\omega-S(AB_1)_\omega=-S(B_1|A)_\omega-S(B_2)_\omega\leq 0
\end{equation} 
because $\omega_{AB_1}$ is classical in $A$. Hence since due to Eq.\ (\ref{39-inter}) the previous bound is saturated by pure states, this one is saturated as well. Otherwise, recall that by theorem \ref{commsq} the relation holds as an equality if there exists a recovery map $\mathcal{R}$ such that
\begin{equation}\label{recoveryPsi}
\mathcal{R}\left(\mc{E}_{RA\id_{B_1}B_2}(\rho)\right)=\rho\qquad
\mathcal{R}\left(\mc{E}_{\tilde{R}\tilde{A}\id_{B_1B_2}} (\rho)\right)=\mc{E}_{\tilde{R}\tilde{A}B_1\id_{B_2}}(\rho)
\end{equation}
hold for this particular $\rho_{RAB_1B_2}$.
If $\rho_{AB_1B_2}=\rho_{AB_1}\otimes \rho_{B_2}$ we may define $\mathcal{R}(\cdot):=U\mathcal{Q}_{\tilde{A}B_1}(\cdot)U^\dagger$, where $U=\sum_{k=1}^{|R|}|r_k\rangle\langle r_k|\otimes e^{-iGr_k}$ and $\mathcal{Q}_{\tilde{A}B_1}(\cdot)$ is the discard and prepare map 
\begin{equation}
\mathcal{Q}_{AB_1}(\sigma_{AB_1C})=\rho_{AB_1}\otimes\sigma_C
\end{equation}
where $C$ is any additional system beyond $AB_1$.
It is straightforward to check that $\mathcal{R}$ indeed satisfies Eq.\ (\ref{recoveryPsi}).
\subsubsection{Significance of the bounds.}
Let us comment on the significance of these bounds for the tripartite game. The right hand side of Eq.\ (\ref{relA}) is always positive, so the relation does in fact pose non trivial bounds on the probability of Bob to win the game, nevertheless it is worth noticing that
\begin{equation}\label{triv}
\kappa_{RA}=U\kappa_R \otimes \rho_AU^\dagger\,,
\end{equation}
with $U=\sum_{k=1}^{|R|}|r_k\rangle\langle r_k| \otimes e^{-iGr_k}$, which is unitary. Hence $S(RA)_\kappa=S(R)_\kappa +S(A)_\rho$. The relation in equation \eqref{relA} reduces to

\begin{equation}
S(AB_1)_\rho+S(AB_2)_\omega \geq S(AB_1)_\omega+S(B_2)_\rho+\max(0,I(A:B_1)_\omega-I(B_1:B_2)_\rho+S(A|B_1B_2)_\rho)\,.
\end{equation}
This is not a trivial bound, but it only involves the pinching map and it is not a statement about the rotation twirl. The problem is the artificial conditioning of the entropy $S(RA)_\kappa$. As a matter of fact, in light of Eq.\ \eqref{triv}, the non trivial contribution of the state $\kappa$ is the conditioning of the entropy. In the next section we will obtain a relation for the bipartite game by trying to make the conditioning of the entropy of the state $\kappa_{RAB_1B_2}$ appear naturally in the inequality.
\subsection{The bipartite game}\label{secondgame}
Let us now consider the second version of the game, this time Bob only has to guess either the rotation or the energy for each round of the game. One expects thus a constraint on the quantity $S(R|AB)_\kappa + S(A|B)_\omega$. To obtain such a relation, let us exploit the property in Eq.\ \eqref{triv} and try to get the term $S(AB)_\kappa$ on the right hand side naturally. Consider the following commuting square 
\begin{equation}
\begin{pmatrix}
A\id_B &\subset &AB\\
\cup & & \cup \\
\tilde{A}\id_B & \subset & \tilde{A}B
\end{pmatrix}
\end{equation}
and start from the state 
\begin{equation}
\kappa_{AB}=\sum_{k=1}^{|R|}p_ke^{-iGr_k}\rho_{AB} e^{iGr_k}\,.
\end{equation}
The state on $\tilde{A}B$ is just $\omega_{AB}$ and the $\log|B|$ terms cancel as always. The relation, keeping the notation $\tilde{A}\rightarrow A$, is
\begin{equation}\label{nontriv}
S(A)_\kappa+S(AB)_\omega\geq S(AB)_\kappa+S(A)_\omega\,.
\end{equation}
One can immediately see that this is a non trivial relation involving both the state $\kappa$ and the state $\omega$. We can now add $S(R)_\kappa+S(AB)_\rho$ on both sides, use Eq.\ \eqref{triv}, and subtract $S(B)_\rho$ to get
\begin{equation}
S(R|AB)_\kappa+S(A|B)_\omega\geq S(R)_\kappa+S(A|B)_\rho+S(A)_\omega-S(A)_\kappa\,.
\end{equation}
Using \eqref{relAE}, this can be rewritten  as 
\begin{equation}
S(R|AB)_\kappa+S(A|B)_\omega\geq S(R)_\kappa+ D(\kappa_A||\omega_A)+ S(A|B)_\rho\,.
\end{equation}
 Equality holds if \eqref{nontriv} takes equality, and this by theorem \ref{commsq} holds if there exists a recovery map
\begin{equation}
\mathcal{R}(\mathcal{E}_{\tilde{A}I_B})=\mathcal{E}_{AI_B}(\rho_{AB})\qquad
\mathcal{R}(\mathcal{E}_{\tilde{A}B})=\rho_{AB}
\end{equation}
If $\rho_{AB}$ is a product state we may simply take $\mc{R}$ to be $\mc{Q}_A$, the operation of resetting the state of $A$ to $\rho_A$ just as in section \ref{secrel}. If $\rho_A$ is a pure eigenstate of $G$ clearly the recovery map is the identity, hence in both of these cases the bound is saturated. 
 \section{Numerical calculations}\label{numerical}
 Here we present some explicit numerical results as an example of our bounds.
 \subsection{The tripartite game}\label{secplots1}
Our bound for the tripartite game is given by Eq.\ \eqref{strel},
which is saturated when either $\rho_{AB_1B_2}$ is pure or $\rho_{AB_1B_2}=\rho_{AB_1}\otimes\rho_{B_2}$. Let us restrict for the moment to the case $B_2\simeq B$, $B_1\simeq \mathbb{C}$, then the bound reduces to
  \begin{equation}\label{eur_nonsplit}
 S(R|A)_\kappa+S(A|B)_\omega  \geq S(R)_\kappa+ D(\kappa_{A} || \omega_{A}) + \max\{0,S(A|B)_\rho\}\,.
 \end{equation}
 This is to be compared to the one found in \cite{TEEUR}
   \begin{equation}\label{eur_coles}
 S(R|A)_\kappa+S(A|B)_\omega  \geq S(R)_\kappa+ D(\kappa_{A} || \omega_{A})\,.
 \end{equation}
 We take $|A|=|B|=2$, $|R|=6$ with random angles following a uniform distribution and $G=\sigma_x$. In the following the right and left hand sides of the bounds are computed and compared for
 
 \begin{equation}
 \rho_{AB}=|\psi\rangle\langle\psi|\otimes |\psi\rangle\langle\psi|\,,
 \end{equation}
 with $|\psi\rangle = \cos\frac\theta 2 |0\rangle +\sin\frac\theta 2 |1\rangle$, where $\theta\in[0,\pi]$. This is a pure product state. Random noise is added to either $|\psi\rangle\langle \psi|$ or $\rho$ itself to obtain a mixed product state or a mixed non product state respectively. The random noise is obtained by adding a random state produced by the function \lstinline{rand_dm} from the Python package QuTiP \cite{JOHANSSON20121760} and rescaling to obtain a trace one matrix. In Figure \ref{plots_first_game} the relevant quantities are plotted for the three cases of a pure product state, a mixed product state and a mixed non product state.

 \begin{figure}[t!]
\centering     
\subfigure[{\bf $\rho_{AB_1B_2}=\rho_{AB_1}\otimes\rho_{B_2}$ is a mixed product state.} ]{\label{fig:b}\includegraphics[width=60mm]{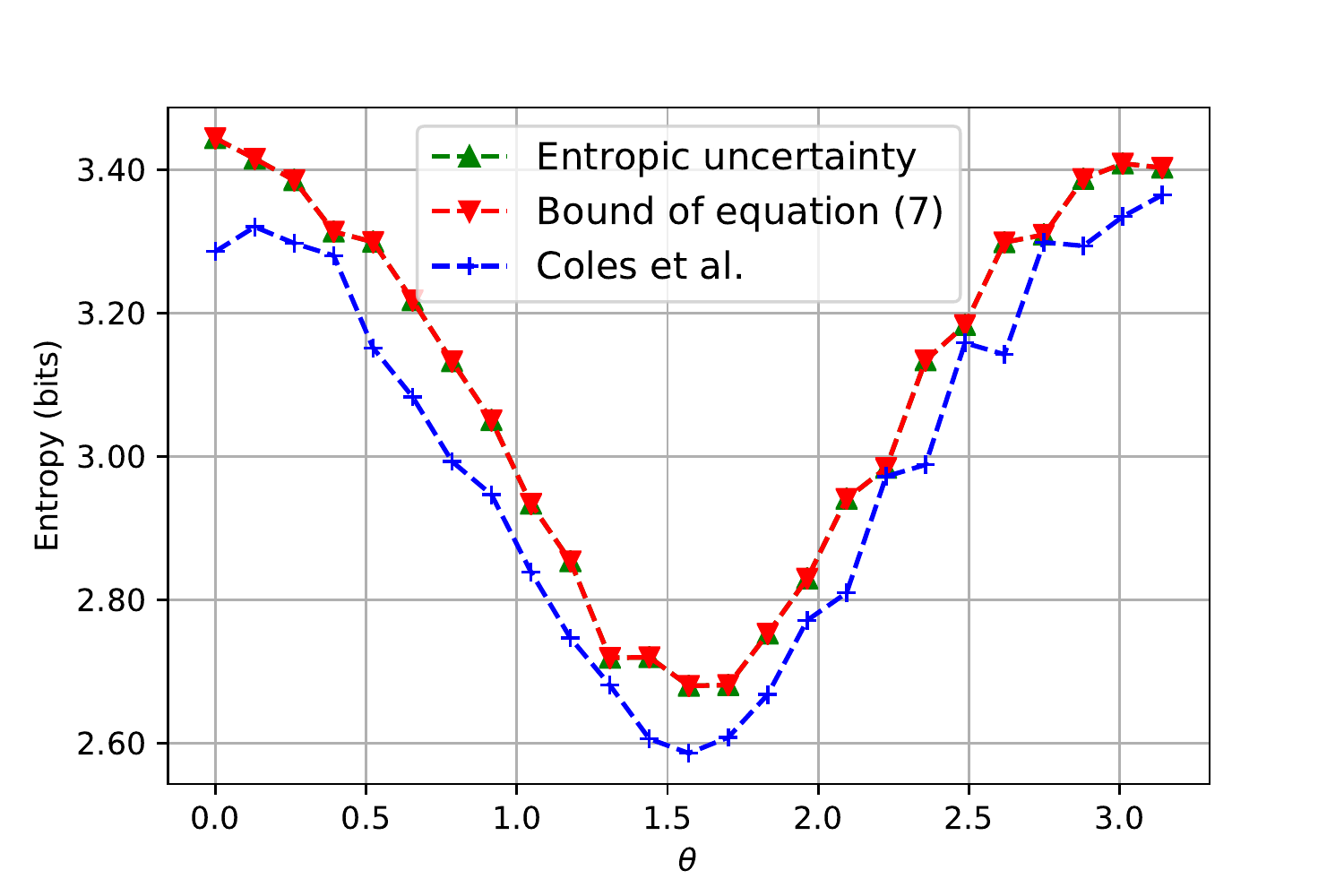}}
\subfigure[{\bf $\rho_{AB_1B_2}$ is a mixed non-product state.}]{\label{fig:b}\includegraphics[width=60mm]{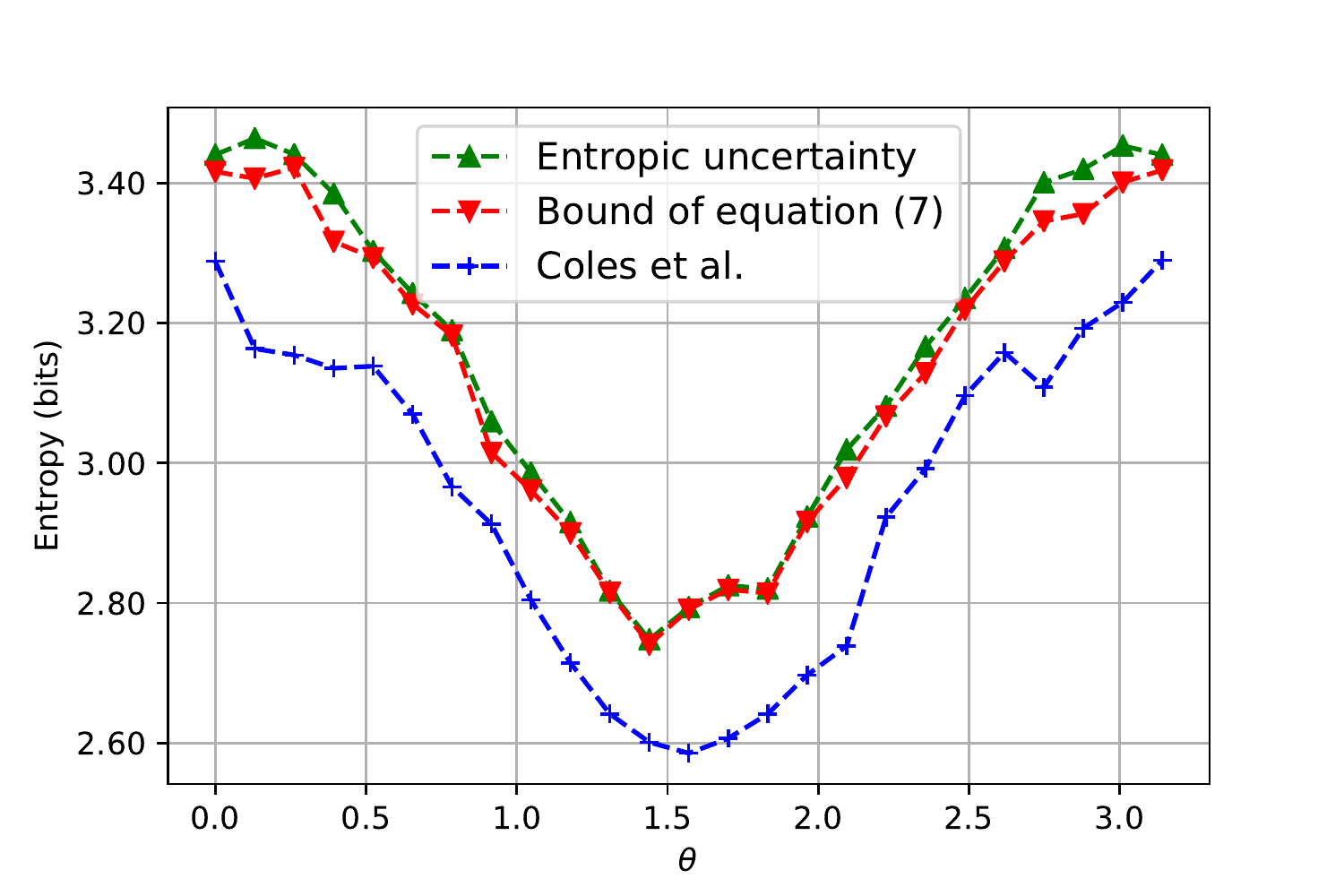}}
 	\caption{{\bf Comparison of bounds for the tripartite game when $B_1$ is trivial.} The above plots compare the bounds obtained by us [cf.\ Eq.\ (\ref{strel})] and by Coles et al. in Ref.\ \cite{TEEUR} for different states $\rho_{AB_1B_2}$. In the case when $\rho_{AB_1B_2}=\rho_{AB_1}\otimes\rho_{B_2}$, our bound (\ref{strel}) is tight whereas the one by Coles et al. \cite[Eq.\ (E10)]{TEEUR} is not. In the generic case, our bound is not tight but still better than the one in Ref.\ \cite{TEEUR}. }
 	\label{plots_first_game}
 \end{figure}
For the tripartite case, where both $B_1$ and $B_2$ are nontrivial, Eq.\ \eqref{strel} is to be compared with the one found in \cite{TEEUR}

\begin{equation}
S(R|AB_1)_\kappa+S(A|B_2)_\omega  \geq \log|R|\,.
\end{equation}
 Note that since in these computations the angles follow a uniform distribution, and thus $S(R)_\kappa=\log|R|$. In Figure \ref{plots_first_game_split} the relevant quantities are plotted taking $|B|=4$, $|B_1|=|B_2|=2$, for the state
 \begin{equation}
 \rho_{AB}=|\psi\rangle\langle\psi|\otimes |\psi\rangle\langle\psi|\otimes |\psi\rangle\langle\psi|
 \end{equation}
 with added random noise.
\begin{figure}
	\centering
	\includegraphics[width=60mm]{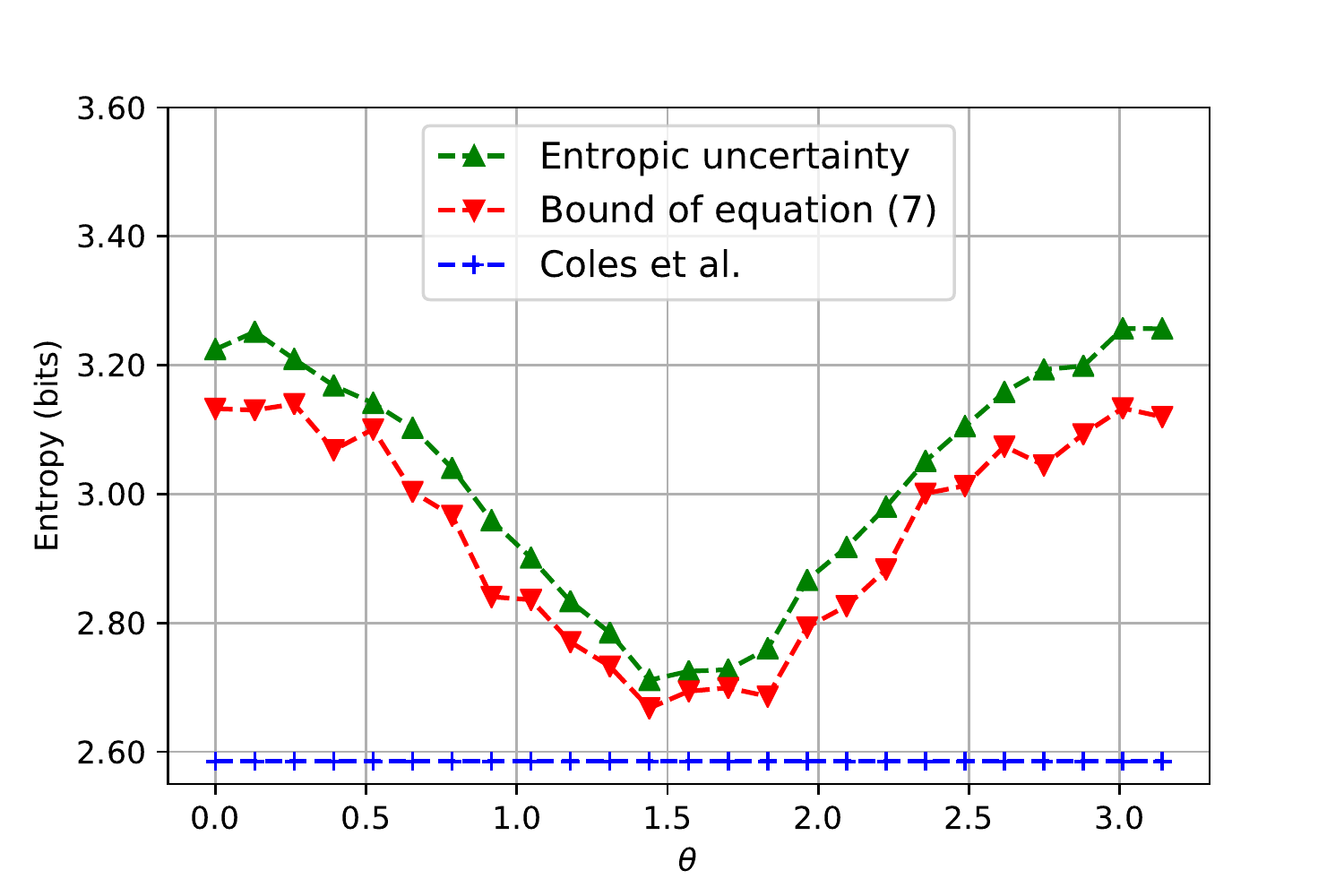}
	\caption{{\bf Comparison of bounds for the first game when $B_1$ is not trivial.} The above plot compares the bounds obtained by us [cf.\ Eq.\ (\ref{strel})] and by Coles et al.\ \cite[Eq.\ (8)]{TEEUR} for generic $\rho_{AB_1B_2}$. Notice that the bound \cite[Eq.\ (8)]{TEEUR} is simply $\log|R|$ and is thus independent of the state's parameter $\theta$. The plot manifests the gap between the entropic uncertainty and the bound by Coles et al., and that our bound is very close to the real uncertainty. }
	\label{plots_first_game_split}
\end{figure}
 \subsection{The bipartite game}
Recall that our bound for the entropic uncertainty in the bipartite game, given by Eq.\ \eqref{eub-second}, is $S(R|AB)_\kappa+S(A|B)_\omega\geq S(R)_\kappa+ D(\kappa_A||\omega_A)+S(A|B)_\rho$.
It is saturated if $\rho_{AB}$ is a product state or if it is a pure eigenstate of $G$. In Figure \ref{second_game_plots}, the bound is further tested for generic, non-product states generated in the same random way as in section \ref{secplots1} for $|A|=|B|=2$. It can be seen that the bound is still considerably, though not rigorously, tight for generic states.

\begin{figure}[t!]
	\centering     
	\includegraphics[width=60mm]{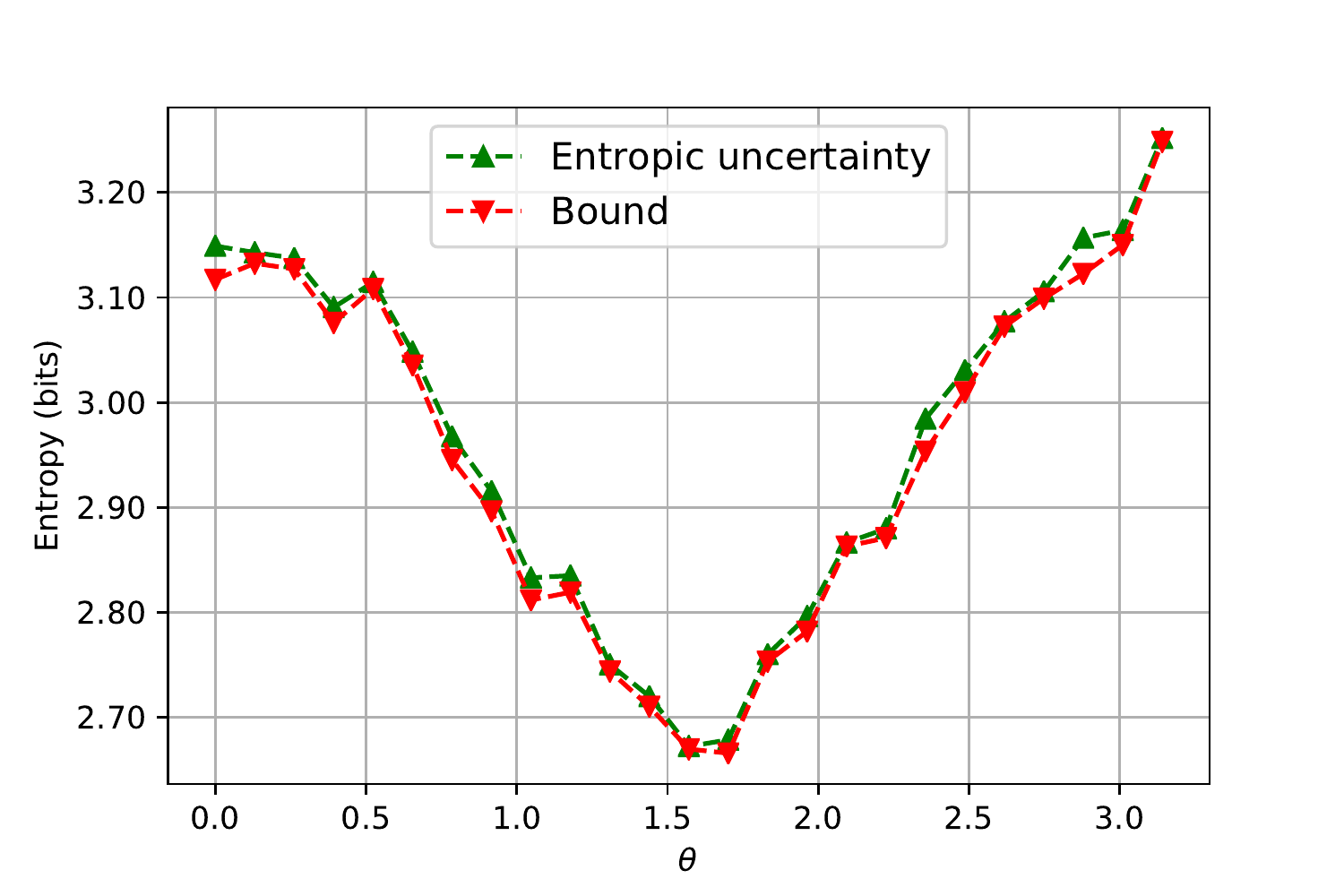}
	\caption{{\bf Performance of the entropic uncertainty bound for the bipartite game.} In this plot, we examine the tightness of our bound (\ref{eub-second}) for the bipartite guessing game. It can be seen that our bound is very close to the true value of the uncertainty, even for generic, non-product states.}\label{second_game_plots}
\end{figure}
\section{Conclusions}\label{conclusions}

 In this work, we utilized the commuting square framework to derive time-energy entropic uncertainty relations based on two different guessing games. 
Our bound for the tripartite game tightens a previous bound in Ref.\ \cite{TEEUR}, in a way similar to other improvements \cite{tomamichel2012framework,adabi2016tightening} made to the standard entropic uncertainty bound.
Our bounds also strengthen the understanding of time-energy uncertainty, by showing that there is a fundamental difference between the case where the quantum memory is split between two parties and the case where one party holds the whole quantum memory. More precisely, the former case renders a game that is impossible to win, while the latter corresponds to a game that is possible to win but only with quantum memory.

Our work demonstrates the power of the algebraic approach, which can also be applied to derive other entropic uncertainties. It remains open, however, how to extend our result to generic R\'enyi entropies. Some hints have already been given in Ref.\ \cite{SSA}, but it might still require a considerable amount of effort to generalize the algebraic approach to this more general setting.


\begin{acknowledgments}
This work was supported by the Swiss National Science Foundation (SNSF) via the National Centre of Competence in Research “QSIT”, as well as the Air Force Office of Scientific Research (AFOSR) via grant  FA9550-19-1-0202. We thank Alessandra Ortelli for drawing Figure \ref{drawing_first_game} and Figure \ref{drawing_second_game}. 
\end{acknowledgments}

\bibliography{bibliography}

\end{document}